\documentclass[12pt]{article}
\usepackage[T1]{fontenc}
\usepackage{graphicx}
\usepackage[latin9]{inputenc}
\usepackage{amsmath,mathrsfs,mathtools,amsthm}
\usepackage{amsmath,amsfonts,amssymb,epsfig,multicol,multirow}
\usepackage{rotating,hhline,stmaryrd}
\usepackage{float}
\usepackage{array}

\usepackage{nicematrix} 

\newtheorem{theorem}{Theorem}
\newtheorem{proposition}{Proposition}
\newtheorem{lem}{Lemma}
\newtheorem{cor}{Corollary}
\newtheorem{remark}{Remark}
\newtheorem{example}{Example}

\newcommand{\ben}{\begin{equation*}}
\newcommand{\een}{\end{equation*}}

\newcommand{\comment}[1]{}

\renewcommand{\arraystretch}{0.7}

\begin{document}

\title{Binary optimal linear codes with various hull dimensions and entanglement-assisted QECC}

\author{Jon-Lark Kim\thanks{Department of Mathematics,
Sogang University,
Seoul 04107, South Korea.
{Email: \tt jlkim@sogang.ac.kr},
J.-L. Kim was supported by the National Research Foundation of Korea (NRF) grant funded by the Korea government (NRF-2019R1A2C1088676).
}
}

\date{}

\maketitle

\begin{abstract}
The hull of a linear code $C$ is the intersection of $C$ with its dual. To the best of our knowledge, there are very few constructions of binary linear codes with the hull dimension $\ge 2$ except for self-orthogonal codes.
We propose a building-up construction to obtain a plenty of binary  $[n+2, k+1]$ codes with hull dimension $\ell, \ell +1$, or $\ell +2$ from a given binary $[n,k]$ code with hull dimension $\ell$. In particular, with respect to hull dimensions 1 and 2, we construct all binary optimal $[n, k]$ codes of lengths up to 13. With respect to hull dimensions 3, 4, and 5, we construct all binary optimal $[n,k]$ codes of lengths up to 12 and the best possible minimum distances of $[13,k]$
codes for $3 \le k \le 10$. As an application, we apply our binary optimal codes with a given hull dimension to construct several entanglement-assisted quantum error-correcting codes(EAQECC) with the best known parameters.
\end{abstract}

\noindent
{\bf keywords} building-up construction, codes, hull, LCD codes


\section{Introduction}

The hull of a linear code $C$ is the intersection of $C$ with its dual.
The hull of a linear code was introduced by Assmus, Jr.
and Key~\cite{AssKey}. The hull determines
the complexity of algorithms for checking permutation equivalence of two linear codes~\cite{Sen2000}, which are very
effective if the dimension of the hull is small, and which are worst if the dimension of the hull is large.
The hardness of the Permutation Code Equivalence problem is of
great importance when designing cryptographic primitives, such as public-key
cryptosystems and identification schemes in the field of code-based cryptography~\cite{SenSim}.

When the hull contains only the zero vector, that is, the hull dimension is 0, $C$ is called a Linear Complementary Dual code (shortly, LCD code). Recently, LCD codes have been actively studied due to its side channel attack. An LCD code was originally constructed by Massey~\cite{Massey},~\cite{Massey92} as a reversible code in order to provide an optimum linear coding solution for the two-user binary adder channel. Carlet and Guilley~\cite{CarGui} introduced several constructions of LCD codes and investigated an application of LCD codes against side-channel attacks(SCA) and Fault Injection Attack(FIA).


There are several constructions for binary linear codes with hull dimensions 0 or 1. More precisely,  Galvez et al.~\cite{Gal} have constructed all binary optimal LCD $[n,k]$ codes for $1 \le k \le n \le 12.$ Harada and Saito~\cite{HarSai} have extended this for $1 \le k \le n \le 16$. Li and Zeng~\cite{LiZeng} have constructed binary linear $[n, k]$ codes with hull dimension 1 for $n=8$ with $k=3,5,7$, $n=9$ with $k=3,5,6,7$, and $n=10$ with $k=3,4,7$, whose optimality was not discussed.

On the other hand, when the dimension $h$ of the hull of a linear $[n,k]$ code $C$ is equal to $k$, $C$ is called self-orthogonal, and self-dual if $h=k=n/2$. Self-orthogonal or self-dual codes have been one of the most active research areas in classical coding theory~\cite{RaiSlo} and recently in quantum coding theory~\cite{CRSS}.
As far as we know, there are few constructions of binary linear codes with the hull dimension $h \ge 2$ except for self-orthogonal codes. It turns out that linear codes with various hull dimensions can be used to construct entanglement-assisted quantum error-correcting codes (EAQECC)~\cite{GalHer}, \cite{GueJitGul}, \cite{SokQia}.

Therefore, it is an interesting problem to find a unified method to construct linear codes with various hull dimensions.



In this paper, we give an efficient and systematic method, called a building-up construction to construct linear codes with various hull dimensions from a given linear code with a fixed hull dimension. More precisely, with respect to hull dimensions 1 and 2, we construct all binary optimal $[n, k]$ codes of lengths up to 13. With respect to hull dimensions 3, 4, and 5, we construct all binary optimal $[n,k]$ codes of lengths up to 12 and the best possible minimum distances of $[13,k]$ codes for $3 \le k \le 10$.
As a coding theoretical application, given length $2 \le n \le 12$ and dimension $k$ $(2 \le k \le n)$,  by running all values of the hull dimension $h$ $(0 \le h \le [n/2])$, we can recover all the binary best known linear $[n,k]$ codes in Grassl's table~\cite{Gra} with sometimes more than one inequivalent code. We apply our binary optimal codes with a given hull dimension to the construction of $[[n,k,d; c]]$ EAQECC with the best known parameters as described in~\cite{GalHer},~\cite{SokQia}.

\medskip

\section{Preliminaries}

A {\em linear $[n,k,d]$ code} $C$ over $GF(q)$ or $\mathbb F_q$ is a $k$-dimensional subspace of $\mathbb F_2^n$ with minimum distance $d(C)$ or $d$ if there is no confusion. The {\em dual} of $C$ is $C^{\perp}=\{{\bf x} \in \mathbb F_2^n~|~ {\bf x} \cdot {\bf c} = 0 {\mbox{ for any }} {\bf c} \in C\}$, where the dot product is the usual inner product. A linear code $C$ is called {\em self-orthogonal} if $C \subset C^{\perp}$ and {\em self-dual} if $C=C^{\perp}$. A linear code $C$ is called an {\em LCD code} (linear complementary dual code) if $C \cap C^{\perp} = \{0 \}$. Hence being LCD is the opposite concept of self-orthogonality.

Let $C$ be a linear code over $GF(q)$ with its dual $C^{\perp}$. The {\em Hull} of $C$ is defined as Hull$(C)= C \cap C^{\perp}$. Let $h={\mbox{dimension of Hull}}(C)$.

We call $C$ {\em $h_i$-optimal} if $d(C)$ is the largest among all the linear $[n,k]$ codes $C$ with $h=i$ for $0 \le i \le k$. We call $C$ {\em optimal} if $d(C)$ is the largest among all the linear $[n,k]$ codes $C$.

\begin{lem}\label{lem_rank}{\rm{(\cite[Proposition 1]{LiZeng})}}
Let $C$ be a linear $[n, k]$ code over $GF(q)$ with generator matrix $G$. Then
$h = k - {\mbox{rank}} (GG^T)$.
\end{lem}

Hence, if $h=0$, that is, Hull$(C)= \{0\}$ or ${\mbox{rank}} (GG^T)=k$, then $C$ is {\em LCD}. If Hull$(C)=C$, then $C$ is self-orthogonal.

\medskip

Now we also describe entanglement-assisted quantum error-correcting codes(EAQECC).
An {\em EAQECC with parameters $[[n, k, d; c]]$} encodes $k$
logical qubits into $n$ physical qubits with the help of $c$ pre-shared entanglement pairs~\cite{QiaZha}. If $c=0$, then
$[[n, k, d; c]]$ EAQECC are equivalent to quantum stabilizer codes. Hence, $[[n, k, d; c]]$ EAQECC are a generalization of $[[n, k, d]]$ QECC.

The following is a useful method to construct $[[n, k, d; c]]_q$ EAQECC from $q$-ary linear $[n, k, d]$ codes.

\begin{proposition} {\rm (\cite{GalHer},~\cite[Corollary 3.1]{GueJitGul},~\cite[Proposition 8]{SokQia})} \label{prop-EAQECC}
 Let $C$ be a linear code over $GF(q)$ with parameters $[n,k,d]_q$ and $C^{\perp}$ be its
dual with parameters $[n,n-k,d']_q$. Let $\dim({\rm{Hull}}(C))=h$. Then, there exist
an $[[n,k-h,d; n-k-h]]_q$ EAQECC and an $[[n,n-k-h, d';k-h]]_q$ EAQECC.
\end{proposition}

\begin{proposition}[\cite{GueJitGul}] \label{prop-quan-MDS}
An $[[n,k,d;c]]_q$ EAQECC satisfies
\[
n+c-k \ge 2(d-1),
\]
where $0 \le c \le n-1$.
\end{proposition}

An EAQECC attaining this Singleton bound is called an {\em MDS EAQECC}. Chen et al.~\cite{CheZhuKai} constructed
MDS EAQECC when $q=2^e$ with $e$ odd with special values of $n, k, d,$ and $c$.

Let $q=2$ and consider the binary Hamming $[7,4,3]$ code $\mathcal H_3$. Since its dual $\mathcal H_3^{\perp}$ is the simplex code $\mathcal S_3$ and is a subcode of $\mathcal H_3$, we have $h(\mathcal H_3)=3.$ Hence by Proposition~\ref{prop-EAQECC}, we obtain a $[[7, 1, 3; 0]]_2$ EAQECC which is best known by Grassl's table~\cite{Gra}. Note that $n+c-k=7+0-1=6$ and $2(d-1)=4$. Hence $[[7, 1, 3; 0]]_2$ EAQECC is not MDS.

In this paper, we consider $q=2$ and construct various
$[[n,k-h,d; n-k-h]]_2$ EAQECC using the building-up constructions.

\section{Building-up construction for linear codes with various hull dimensions}

In the remaining sections, we consider binary codes.
We can construct $[n+2, k+1]$ linear codes with hull dimension $\ell +1$ from a given $[n, k]$ linear code with hull dimension $\ell$ as follows.

\begin{theorem} \label{build_2}
Let $C$ be a binary linear $[n,k]$ code. Suppose that the dimension of Hull$(C)$ is $\ell$, where $0 \le \ell \le k$. Let $G$ be a generator matrix for $C$ and $H$ a parity check matrix for $C$.

 Suppose that ${\bf x}=(x_1, x_2, \dots, x_n) \in GF(2)^n$ satisfies ${\bf x} \cdot {\bf x} =1$.
Let $y_i = {\bf x} \cdot {\bf r}_i$ for $1 \le i \le k$ where ${\bf r}_i$ is the $i$th row of $G$ and
$z_j = {\bf x} \cdot {\bf s}_j$ for $1 \le j \le n-k$ where ${\bf s}_j$ is the $j$th row of $H$.
Then
\begin{enumerate}

\item[(a)]

the following matrix

\[
G_1 = \left[
\begin{array}{cc|ccc}
1   & 0   & x_1 & \dots & x_n \\
\hline
y_1 & y_1 &     &   {\bf r}_1 &     \\
y_{2} & y_{2} &  & {\bf r}_2  &   \\
\vdots & \vdots & & \vdots    &   \\
y_k & y_k &     &  {\bf r}_k  &     \\
\end{array}
\right]
\]
generates an $[n+2, k+1]$ linear code $C_1$ with $h(C_1)=\ell +1$. This is called {\bf Construction I}.

\item[(b)]
A parity check matrix $H_1$ for $C_1$ is given by
\[
H_1 = \left[
\begin{array}{cc|ccc}
1   & 0   & x_1 & \dots & x_n \\
\hline
z_1 & z_1 &     & {\bf s}_1    &     \\
z_{2} & z_{2} &  &  {\bf s}_2 &   \\
\vdots & \vdots &       & \vdots &   \\
z_{n-k} & z_{n-k} &     &  {\bf s}_{n-k}  &     \\
\end{array}
\right].
\]
\end{enumerate}

\end{theorem}

\begin{proof}
We prove (a). Now rank$(G_1G_1^T)$ is computed as follows.

\[
G_1G_1^T = \left[
\begin{array}{cccc}
0 & 0 & \dots & 0 \\
0 &   &       &  \\
\vdots &  &  GG^T \\
0 &    &  &     \\
\end{array}
\right].
\]

Thus rank$(G_1G_1^T) = {\mbox{rank}}(GG^T) =k - h(C)=k-\ell$ since $h(C)=\ell$. Now $h(C_1)=(k+1)- {\mbox{rank}}(G_1G_1^T)=(k+1)-(k -\ell)=\ell + 1$ as desired.

\medskip

We prove (b) as follows.
Notice that
\setlength{\tabcolsep}{4pt}
\[
G_1 H_1^T =
\left[
\begin{array}{cccc}
0  & 0 & \dots & 0 \\
0  &   &       &   \\
\vdots &    & GH^T \\
0  &   &    &      \\
\end{array}
\right] = 0.
\]
Since the top row of $H_1$ cannot be a linear combination of the remaining rows of $H_1$, the dimension of the row space of $H_1$ is $1+(n-k)=(n+2)-(k+1)$ which is the dimension of $(C_1^{\perp})$. Thus $H_1$ is a generator matrix for $C_1^{\perp}$.
\end{proof}

\medskip

\begin{theorem} \label{build_2_2}
Let $C$ be a binary linear $[n,k]$ code. Suppose that the dimension of Hull$(C)$ is $\ell$, where $0 \le \ell \le k$. Let $G$ be a generator matrix for $C$ and $H$ a parity check matrix for $C$.

Suppose that ${\bf x}=(x_1, x_2, \dots, x_n) \in GF(2)^n$ satisfies ${\bf x} \cdot {\bf x} =0$.
Let $y_i = {\bf x} \cdot {\bf r}_i$ for $1 \le i \le k$ where ${\bf r}_i$ is the $i$th row of $G$ and
$z_j = {\bf x} \cdot {\bf s}_j$ for $1 \le j \le n-k$ where ${\bf s}_j$ is the $j$th row of $H$.
Then
\begin{enumerate}

\item[(a)]

the following matrix
\[
G_2 = \left[
\begin{array}{cc|ccc}
1   & 1   & x_1 & \dots & x_n \\
\hline
y_1 & 0 &     &  {\bf r}_1  &     \\
y_{2} & 0 &    &  {\bf r}_2   &  \\
\vdots & \vdots & &  \vdots   &   \\
y_k & 0 &     &  {\bf r}_k  &     \\
\end{array}
\right]
\]
generates an $[n+2, k+1]$ linear code $C_2$ with $h(C_2)=\ell$, $\ell +1$, or $\ell +2$. More precisely, we characterize them as follows.

\begin{itemize}

\item
If $y_i=0$ for any $1 \le i \le k$, then the dimension of $Hull(C_2)$ is $\ell +1$. This is called {\bf Construction II}.

\item Suppose $y_i \neq 0$ for some $1 \le i \le k$. So, $G_2$ can be rewritten as $G_2'$ given by
    \[
G_2' = \left[
\begin{array}{cc|ccc}
1   & 1   & x_1 & \dots & x_n \\
\hline
1 & 0 &     &    &     \\
0 & 0 &     &    &     \\
\vdots & \vdots & & G'     &   \\
0 & 0 &    &   \\
\end{array}
\right],
\]
where $\left< G' \right>=\left< G \right>$.
Then the dimension of $Hull(C_2)$ is $\ell$, $\ell +1$, or $\ell +2$. This is called {\bf Construction III}.
\end{itemize}

\item[(b)]
A parity check matrix $H_2$ for $C_2$ is given by
\[
H_2 = \left[
\begin{array}{cc|ccc}
1   & 1   & x_1 & \dots & x_n \\
\hline
0 & z_1 &     &   {\bf s}_1 &     \\
0 & z_{2} &   &  {\bf s}_2 &   \\
\vdots & \vdots &       & \vdots &   \\
0 & z_{n-k} &     &  {\bf s}_{n-k}  &     \\
\end{array}
\right].
\]

\end{enumerate}

\end{theorem}

\begin{proof}

We prove (a).
\begin{itemize}

\item
Suppose that $y_i= 0$ for any $1 \le i \le k$. Then
we have

\[
G_2G_2^T = \left[
\begin{array}{cccc}
0 & 0 & \dots & 0 \\
0 &   &       &  \\
\vdots &  &  GG^T \\
0 &    &  &     \\
\end{array}
\right].
\]

Thus rank$(G_2G_2^T) = {\mbox{rank}}(GG^T) =k - \ell$ since $h(C)= \ell$. The dimension of Hull$(C_2) = h(C_2)=(k+1)- {\mbox{rank}}(G_2G_2^T)=(k+1)-(k-\ell)=\ell +1$.

\item

Suppose $y_i \neq 0$ for some $1 \le i \le k$. By row operations of $G_2$, $G_2$ is transformed into $G_2'$ given by
    \[
G_2' = \left[
\begin{array}{cc|ccc}
1   & 1   & x_1 & \dots & x_n \\
\hline
1 & 0 &     &    &     \\
0 & 0 &     &    &     \\
\vdots & \vdots & & G'     &   \\
0 & 0 &    &   \\
\end{array}
\right],
\]
where $\left< G' \right>=\left< G \right>$. Furthermore,

\[
G_2'(G_2')^T =
\left[
\begin{array}{c|ccc}
0   & 0   & \dots & 0 \\
\hline
0 &  &     &         \\
\vdots & &  X     &   \\
0 &  &    &   \\
\end{array}
\right],
\]
where $X=(1 0 \dots 0)^T(1 0 \dots 0) + G'(G')^T$.
Since rank$(A+B) \le$ rank$(A)$ + rank$(B)$, we have
rank$(X)$ $\le$ rank$((1 0 \dots 0)^T(1 0 \dots 0))$ + rank$(G'(G')^T)$ $=1 +$ rank$(G'(G')^T)$.
Noting that $(1 0 \dots 0)^T(1 0 \dots 0)$ affects only the top row of $G'(G')^T$, we know that rank$(X)$ decreases by at most one. So, rank$(X)$ is rank$(G'(G')^T)$, rank$(G'(G')^T) -1$, or rank$(G'(G')^T) +1$.

Thus rank$(G_2 (G_2)^T)$=rank$(G_2' (G_2')^T)$ = rank$(X)$ is  rank$(G'(G')^T)$ = $k -\ell$, rank$(G'(G')^T) -1= k -\ell -1$, or rank$(G'(G')^T) +1= k -\ell +1$. Hence the dimension of $Hull(C_2)$ is $k+1 -(k -\ell)=\ell +1$, $k+1 -(k - \ell -1)=\ell +2$, or $k+1 -(k - \ell +1)=\ell$.
\end{itemize}

We prove (b). It is straightforward to check that $G_2H_2^T = 0$ by the definition of $y_i$'s and $z_j$'s. Because the rank of $H_2$ is $(n-k) +1$ and the dimension of the dual of $C_2$ is $(n+2)-(k+1)=n-k+1$, we see that $H_2$ is a parity check matrix for $C_2$.
\end{proof}

We note that Constructions II and III are basically the same construction but we distinguish them in order to guess the hull dimension of the built-up code.
We remark that Theorem~\ref{build_2} reproves the original building-up construction of binary self-dual codes~\cite{Kim01} where $n$ is even and $k=n/2=\ell$.

\medskip

\medskip

Harada~\cite{Har2021} gave a construction of binary LCD $[n+2, k+1]$ codes from a given binary LCD $[n, k]$ code. We generalize this in the following theorem.

By modifying the proof of Theorem~\ref{build_2}, we can construct $[n+2, k+1]$ linear codes with the same hull dimension as that of a given $[n, k]$ linear code.

\begin{theorem}\label{build-k}
Let $C$ be a binary linear $[n,k]$ code. Suppose that the dimension of Hull$(C)$ is $\ell$, where $0 \le \ell \le k$. Let $G$ be a generator matrix for $C$ and $H$ a parity check matrix for $C$.

Suppose that ${\bf x}=(x_1, x_2, \dots, x_n) \in GF(2)^n$ satisfies ${\bf x} \cdot {\bf x} =0$.
Let $y_i = {\bf x} \cdot {\bf r}_i$ for $1 \le i \le k$ where ${\bf r}_i$ is the $i$th row of $G$ and
$z_j = {\bf x} \cdot {\bf s}_j$ for $1 \le j \le n-k$ where ${\bf s}_j$ is the $j$th row of $H$.
The following matrix
\[
G_3 = \left[
\begin{array}{cc|ccc}
1   & 0   & x_1 & \dots & x_n \\
\hline
y_1 & y_1 &     &  {\bf r}_1  &     \\
y_{2} & y_{2} &  & {\bf r}_2  &   \\
\vdots & \vdots & &  \vdots    &   \\
y_k & y_k &     &   {\bf r}_k &     \\
\end{array}
\right]
\]
generates an $[n+2, k+1]$ linear code $C_3$ with $h(C_3)=\ell$. This is called {\bf Construction IV}.
A parity check matrix $H_3$ for $C_3$ is given by
\[
H_3 = \left[
\begin{array}{cc|ccc}
0   & 1   & x_1 & \dots & x_n \\
\hline
z_1 & z_1 &     &  {\bf s}_1  &     \\
z_{2} & z_{2} &   & {\bf s}_2 &   \\
\vdots & \vdots &       & \vdots &   \\
z_{n-k} & z_{n-k} &     & {\bf s}_{n-k}   &     \\
\end{array}
\right].
\]
\end{theorem}

\begin{proof}
The proof is almost the same as that of Theorem~\ref{build_2}. It is straightforward to see that
$G_3H_3^T =0$ and rank$(H_3)=1+(n-k)$ which implies that $H_3$ is a parity check matrix for $C_3$.
We compute rank$(G_3G_3^T)$ as follows.

\renewcommand{\arraystretch}{0.6}

\[
G_3G_3^T =
\left[
\begin{array}{cccc}
1 & 0 & \dots & 0 \\
0 &   &       &  \\
\vdots &  &  GG^T \\
0 &    &  &     \\
\end{array}
\right].
\]
Thus rank$(G_3G_3^T) = 1+{\mbox{rank}}(GG^T) = 1 + k - h(C)=1+ k-\ell$ since $h(C)=\ell$. Now $h(C_3)=(k+1)- {\mbox{rank}}(G_3G_3^T)=(k+1)-(1+ k -\ell)=\ell$ as desired.
\end{proof}

We can estimate the minimum distance $d(C_i)$ $(i=1,2,3)$ for Constructions I-IV as follows.

\begin{theorem}
 Let $C$ be a binary linear $[n,k]$ code. Let ${\bf x}=(x_1, x_2, \dots, x_n) \in GF(2)^n$. Then we have the following.

\begin{enumerate}

\item[{(i)}] The minimum distance $d(C_i)$ ($i=1,3$) is
 $\min\{d(C), {\mbox{{\rm weight}}}({\bf x} + C)+1 \}$ or
$\min\{d(C)+2, {\mbox{{\rm weight}}}({\bf x} + C)+1 \}$ for Constructions I and IV.

\item[{(ii)}]
 $d(C_2)$ is $\min\{d(C), {\mbox{{\rm weight}}}({\bf x} + C)+2 \}$ for Construction II and
$\min\{d(C)+1, {\mbox{{\rm weight}}}({\bf x} + C)+1 \}$, or $\min\{d(C)+1, {\mbox{{\rm weight}}}({\bf x} + C)+2 \}$ for Construction III.
\end{enumerate}

\end{theorem}

\begin{proof}
We prove statement (i). Let $i=1,3$.
Let $S_1$ be the code spanned by all the rows of $G_i$ except for the top row.
Then $C_i$ is the disjoint union of $S_1$ and $(1 ~0~ {\bf x})+S_1$.
Thus $d(C_i)$ is the minimum of $d(S_1)$ and
weight($(1 ~0 ~{\bf x})+S_1$). We note that $d(S_1)$ is $d(C)$ or $d(C)+2$ and that
weight($(1 ~ 0 ~{\bf x})+S_1$)= 1+ weight(${\bf x} + C)$). Hence we obtain (i).
Similarly, we can prove (ii), whose proof is omitted.
\end{proof}

The following is straightforward since $d(C) \le \rho(C)$, where $\rho(C)$ is the covering radius of $C$.

\begin{cor}
 Let $C$ be a binary linear $[n,k]$ code with its covering radius $\rho(C)$. Then
 the minimum distance $d(C_i)$ ($i=1,2,3$) for Constructions I-IV satisfies
$\min\{d(C), {\mbox{{\rm weight}}}({\bf x} + C)+1 \}$ $\le d(C_i) \le \rho(C)+2$.
\end{cor}

\begin{remark}
{\em
Although we have considered Constructions I-IV for linear codes over $GF(2)$, it is easy to see that the same Constructions I-IV in Theorems 1-3 hold for linear codes over $GF(q)$, where $q=2^r$ for any integer $r \ge 1$. If $q$ is odd, then a slight modification of Constructions I-IV based on the building-up construction for self-dual codes over $GF(q)$~\cite{KimLee1},\cite{KimLee2} will give results similar to Theorems 1-3.
}
\end{remark}

\section{Some interesting optimal linear codes}
We display some interesting optimal linear codes with hull dimensions 2 and 3 from a linear code with hull dimension 1.



Start with an $h_1$-optimal $[10,6,3]$ code $C$ whose generator matrix $G$ is given below.

\[
G =
\begin{bNiceMatrix}
1 0 0 0 0 0 0 1 0 1 \\
0 1 0 0 0 0 1 0 0 1 \\
0 0 1 0 0 0 1 1 1 0 \\
0 0 0 1 0 0 0 1 1 0 \\
0 0 0 0 1 0 1 0 1 0 \\
0 0 0 0 0 1 1 1 0 0 \\
\end{bNiceMatrix}
\]

\begin{example}
{\rm
Let us take ${\bf x}=(0 0 0 0 0 1 1 0 0 0)$ with $G$ by Construction III to get an $h_2$-optimal $[12, 7, 3]$ linear code $C'$ with $h=2$. Its generator matrix $G'$ is written as follows.

\[
G' =
\left[
\begin{array}{c|c}
1 1 & 0 0 0 0 0 1 1 0 0 0 \\
\hline
0 0 & 1 0 0 0 0 0 0 1 0 1 \\
1 0 & 0 1 0 0 0 0 1 0 0 1 \\
1 0 & 0 0 1 0 0 0 1 1 1 0 \\
0 0 & 0 0 0 1 0 0 0 1 1 0 \\
1 0 & 0 0 0 0 1 0 1 0 1 0 \\
0 0 & 0 0 0 0 0 1 1 1 0 0 \\
\end{array}
\right]
\sim
\begin{bNiceMatrix}
1 0 0 0 0 0 1 0 1 0 1 0 \\
0 1 0 0 0 0 1 0 1 1 1 0 \\
0 0 1 0 0 0 0 0 0 1 0 1 \\
0 0 0 1 0 0 1 0 0 0 1 1 \\
0 0 0 0 1 0 1 0 0 1 0 0 \\
0 0 0 0 0 1 0 0 0 1 1 0 \\
0 0 0 0 0 0 0 1 1 1 0 0 \\
\end{bNiceMatrix}
\]



By Proposition~\ref{prop-EAQECC}, we can obtain a $[[12, 5, 3; 3]]_2$ EAQECC from $C'$.

}
\end{example}

\begin{example}
{\rm
Let us take ${\bf x}=(1 1 1 1 1 1 0 0 1 1)$ with $G$ above by Construction III to get an optimal $[12, 7, 4]$ linear code $C''$ with $h=3$. Its generator matrix $G''$ is written in standard from after row operations.

\[
G'' =
\left[
\begin{array}{c|c}
1 1 & 1 1 1 1 1 1 0 0 1 1 \\
\hline
0 0 & 1 0 0 0 0 0 0 1 0 1 \\
0 0 & 0 1 0 0 0 0 1 0 0 1 \\
0 0 & 0 0 1 0 0 0 1 1 1 0 \\
0 0 & 0 0 0 1 0 0 0 1 1 0 \\
0 0 & 0 0 0 0 1 0 1 0 1 0 \\
1 0 & 0 0 0 0 0 1 1 1 0 0 \\
\end{array}
\right]
\sim
\begin{bNiceMatrix}
1 0 0 0 0 0 0 1 1 1 0 0 \\
0 1 0 0 0 0 0 0 1 1 0 1 \\
0 0 1 0 0 0 0 1 1 0 0 1 \\
0 0 0 1 0 0 0 1 0 1 0 1 \\
0 0 0 0 1 0 0 0 1 1 1 0 \\
0 0 0 0 0 1 0 1 1 0 1 0 \\
0 0 0 0 0 0 1 1 0 1 1 0 \\
\end{bNiceMatrix}
\]


By Proposition~\ref{prop-EAQECC}, we can obtain a $[[12, 5, 3; 3]]_2$ EAQECC from $C''$.

 By exhaustive search, we have checked that there are up to equivalence exactly two optimal $[12, 7, 4]$ codes. One of them is the above $[12, 7, 4]$ code $C''$ with $h=3$. The other is a $[12, 7, 4]$ code~\cite{KimWeb} with $h=1$ whose weight distribution is
 $[ \left<0, 1\right>, \left<4, 38\right>, \left<6, 52\right>, \left<8, 33\right>, \left<10, 4\right> ]$.
  This code gives a $[[12, 6, 4; 4]]_2$ EAQECC by Proposition~\ref{prop-EAQECC}.
}

\end{example}


\section{Optimal linear codes with several hulls and the construction of EAQECC}

We construct several optimal linear codes of lengths up to 13 with $h=i$ ($i=1,2,3,4,5$) from a given linear code of a fixed hull dimension $h$.
Tables~1,3,5,7,9 display best possible minimum distances of linear $[n,k]$ codes from hull dimensions 1 to 5.
The upper bounds for the minimum distances in the tables are from Grassl's table~\cite{Gra} by taking not the hull dimension into account and by brute force search.
Each cell in each table denotes the highest minimum distance $d(n,k)$ for given $n, k$, and $h=i$ together with the superscripts referring to Constructions I to IV and $o$ meaning that the codes are optimal.
For $n=12$, we apply Constructions I, II, and/or III. For $n=13$, we apply Constructions I, III, and/or IV. All computations were done by Magma~\cite{Bos}.
To save the space, we post whole information about the codes in Tables~1,3,5,7,9 in the author's website~\cite{KimWeb} and list most generator matrices for $n=12$ and $13$ in this paper.

Tables~2,4,6,8,10 display associated $[[n, k, d; c]]_2$ EAQECC based on Proposition~\ref{prop-EAQECC} and Tables 1, 3, 5, 7, 9. In other words, we obtain $[[n, k-h, d; n-k-h]]_2$ EAQECC from binary $[n,k,d]$ codes with hull dimension $h$.

\begin{example}
{\rm
Fix the hull dimension $h=1$.
For any $n$ with $k$ such that $1 \le k \le n \le 11$ and $n=12$ with $k~(1 \le k \le 4)$,  we ran exhaustive search to get optimal or $h_1$-optimal codes. We note that there is an optimal $[12, 5, 4]$ code with $h=1$ from Magma database.

For $n=12$ with $k \ge 6$, we apply Constructions I and II to all the LCD codes of length $10$ and dimension $k-1$ displayed in~\cite{HarSai}.
 More precisely, we construct optimal $[12, 6, 4]$ and $[12, 9, 2]$ codes by Construction I. Similarly, we construct optimal $[12, 7,4]$ and $[12, 8, 3]$ codes by Construction III.

Let $n=13$. We construct
optimal $[13, 4, 6]$, $[13, 5,5]$, $[13, 6,4]$, $[13, 7, 4]$, $[13, 8, 3]$, $[13, 10, 2]$, $[13, 11, 2]$ codes by Construction I.
We also construct an $h_1$-optimal $[13, 3, 6]$ code by Construction I, which is justified by the non-existence of $[13, 3, 7]$ codes with $h=1$ using exhaustive search.
Similarly, we construct an $h_1$-optimal $[13, 9, 2]$ code by Construction I, which is justified by the non-existence of $[13, 9, 3]$ codes with $h=1$ using exhaustive search.
}
\end{example}

\medskip
In what follows, $G_{n,k,d}^{i}$ refers to a generator matrix for a binary $[n,k,d]$ code $C_{n,k,d}^{i}$ with $h=i$ and the highest minimum distance $d=d(n,k)$.

\medskip

\renewcommand{\arraystretch}{0.6}

$\bullet$ $n=12$ with $h=1$


\[
G_{12,6,4}^{1}=
\begin{bNiceMatrix}
1 0 1 1 1 1 0 1 1 1 0 0 \\
1 1 1 0 0 0 1 0 0 0 1 0 \\
1 1 0 1 0 0 1 0 0 0 0 1 \\
1 1 0 0 1 0 0 0 0 0 1 1 \\
0 0 0 0 0 1 1 0 0 1 0 1 \\
0 0 0 0 0 0 0 0 1 1 1 1 \\
\end{bNiceMatrix},~
G_{12,7,4}^{1}=
\begin{bNiceMatrix}
1 1 1 1 1 1 0 1 1 0 0 0\\
1 0 1 0 0 0 1 0 0 1 0 0\\
1 0 0 1 0 0 1 0 0 0 1 0\\
1 0 0 0 1 0 0 0 0 1 1 0\\
1 0 0 0 0 1 1 0 0 0 0 1\\
1 0 0 0 0 0 0 1 0 1 0 1\\
1 0 0 0 0 0 0 0 1 0 1 1\\
\end{bNiceMatrix}
\]

\[
G_{12,8,3}^{1}=
\begin{bNiceMatrix}
1 1 1 1 1 1 0 0 0 0 1 1 \\
1 0 1 0 0 0 0 0 0 0 1 1 \\
0 0 0 1 0 0 0 0 0 1 0 1 \\
0 0 0 0 1 0 0 0 0 1 1 0 \\
1 0 0 0 0 1 0 0 0 1 1 1 \\
0 0 0 0 0 0 1 0 0 0 1 1 \\
1 0 0 0 0 0 0 1 0 1 0 1 \\
0 0 0 0 0 0 0 0 1 1 1 1 \\
\end{bNiceMatrix},~
G_{12,9,2}^{1}=
\begin{bNiceMatrix}
1 0 1 0 0 0 0 0 0 0 0 0\\
1 1 1 0 0 0 0 0 0 0 0 1\\
0 0 0 1 0 0 0 0 0 0 0 1\\
0 0 0 0 0 1 0 0 0 0 0 1\\
0 0 0 0 0 0 1 0 0 0 0 1\\
0 0 0 0 0 0 0 1 0 0 0 1\\
0 0 0 0 0 0 0 0 1 0 0 1\\
0 0 0 0 0 0 0 0 0 1 0 1\\
0 0 0 0 0 0 0 0 0 0 1 1\\
\end{bNiceMatrix}
\]

\[
G_{12,10,1}^{1}=
\begin{bNiceMatrix}
1 0 1 0 0 0 0 0 0 0 0 0\\
0 1 0 0 0 0 0 0 0 0 0 0\\
0 0 0 0 1 0 0 0 0 0 0 0\\
0 0 0 0 0 1 0 0 0 0 0 0\\
0 0 0 0 0 0 1 0 0 0 0 0\\
0 0 0 0 0 0 0 1 0 0 0 0\\
0 0 0 0 0 0 0 0 1 0 0 0\\
0 0 0 0 0 0 0 0 0 1 0 0\\
0 0 0 0 0 0 0 0 0 0 1 0\\
0 0 0 0 0 0 0 0 0 0 0 1\\
\end{bNiceMatrix},~
G_{12,11,2}^{1}=
\begin{bNiceMatrix}
1 0 0 0 0 0 0 0 0 0 0 1\\
0 1 0 0 0 0 0 0 0 0 0 1\\
0 0 1 0 0 0 0 0 0 0 0 1\\
0 0 0 1 0 0 0 0 0 0 0 1\\
0 0 0 0 1 0 0 0 0 0 0 1\\
0 0 0 0 0 1 0 0 0 0 0 1\\
0 0 0 0 0 0 1 0 0 0 0 1\\
0 0 0 0 0 0 0 1 0 0 0 1\\
0 0 0 0 0 0 0 0 1 0 0 1\\
0 0 0 0 0 0 0 0 0 1 0 1\\
0 0 0 0 0 0 0 0 0 0 1 1\\
\end{bNiceMatrix}
\]

$\bullet$ $n=13$ with $h=1$

\[
G_{13, 3,6}^{1}=
\begin{bNiceMatrix}
1 0 1 0 1 1 0 0 0 1 1 0 0\\
1 1 1 1 1 0 0 0 1 1 0 0 1\\
0 0 0 0 0 1 1 1 1 1 0 0 1\\
\end{bNiceMatrix},
~
G_{13, 4,6}^{1}=
\begin{bNiceMatrix}
1 0 0 1 1 1 0 1 1 1 1 0 0\\
1 1 1 0 0 1 1 0 1 1 0 1 1\\
0 0 0 1 0 1 1 1 0 1 0 0 1\\
0 0 0 0 1 1 1 1 1 0 0 1 0\\
\end{bNiceMatrix}
\]

\[
G_{13, 5,6}^{1}=
\begin{bNiceMatrix}
1 0 1 1 1 0 1 1 1 1 0 0 0\\
1 1 1 0 0 0 1 1 0 0 1 0 0\\
1 1 0 1 0 0 1 0 1 0 0 1 0\\
1 1 0 0 1 0 1 0 1 0 1 0 1\\
0 0 0 0 0 1 1 1 0 0 0 1 1\\
\end{bNiceMatrix},~
G_{13, 6,4}^{1}=
\begin{bNiceMatrix}
1 0 1 0 0 1 1 0 0 0 0 0 0\\
1 1 1 0 0 0 0 0 1 0 1 1 1\\
0 0 0 1 0 0 0 1 0 1 0 1 0\\
0 0 0 0 1 0 0 1 1 0 0 1 1\\
1 1 0 0 0 1 0 0 0 1 1 1 0\\
1 1 0 0 0 0 1 1 1 0 1 0 1\\
\end{bNiceMatrix}
\]

\[
G_{13, 7,4}^{1}=
\begin{bNiceMatrix}
1 0 1 0 0 1 1 0 0 0 0 0 0\\
1 1 1 0 0 0 0 0 0 1 0 1 1\\
0 0 0 1 0 0 0 1 0 1 1 1 1\\
0 0 0 0 1 0 0 1 0 1 0 0 1\\
1 1 0 0 0 1 0 0 0 1 1 0 1\\
1 1 0 0 0 0 1 1 0 1 0 1 0\\
0 0 0 0 0 0 0 0 1 1 1 1 0\\
\end{bNiceMatrix},
~
G_{13, 8,3}^{1}=
\begin{bNiceMatrix}
1 0 1 0 0 1 1 0 0 0 0 0 0\\
1 1 1 0 0 0 0 0 0 1 0 1 1\\
0 0 0 1 0 0 0 0 0 1 1 0 1\\
0 0 0 0 1 0 0 0 0 0 1 1 0\\
1 1 0 0 0 1 0 0 0 1 1 1 0\\
1 1 0 0 0 0 1 0 0 1 1 1 1\\
0 0 0 0 0 0 0 1 0 0 1 0 1\\
0 0 0 0 0 0 0 0 1 0 0 1 1\\
\end{bNiceMatrix}
\]

\[
G_{13, 9,2}^{1}=
\begin{bNiceMatrix}
1 0 1 0 0 0 0 0 0 0 0 0 0\\
1 1 1 0 0 0 0 0 0 0 0 0 1\\
0 0 0 1 0 0 0 0 0 0 0 0 1\\
0 0 0 0 1 0 0 0 0 0 0 0 1\\
0 0 0 0 0 1 0 0 0 0 0 0 1\\
0 0 0 0 0 0 1 0 0 0 0 0 1\\
0 0 0 0 0 0 0 0 0 1 0 0 1\\
0 0 0 0 0 0 0 0 0 0 1 0 1\\
0 0 0 0 0 0 0 0 0 0 0 1 1\\
\end{bNiceMatrix},
~
G_{13, 10,2}^{1}=
\begin{bNiceMatrix}
1 0 1 0 0 0 0 0 0 0 0 0 0\\
1 1 1 0 0 0 0 0 0 0 0 1 0\\
0 0 0 1 0 0 0 0 0 0 0 1 0\\
0 0 0 0 1 0 0 0 0 0 0 1 1\\
0 0 0 0 0 1 0 0 0 0 0 1 0\\
0 0 0 0 0 0 1 0 0 0 0 1 0\\
0 0 0 0 0 0 0 1 0 0 0 1 0\\
0 0 0 0 0 0 0 0 1 0 0 1 0\\
0 0 0 0 0 0 0 0 0 1 0 1 0\\
0 0 0 0 0 0 0 0 0 0 1 1 0\\
\end{bNiceMatrix}
\]

\[
G_{13, 11,2}^{1}=
\begin{bNiceMatrix}
1 0 1 0 0 0 0 0 0 0 0 0 0\\
1 1 1 0 0 0 0 0 0 0 0 0 1\\
0 0 0 1 0 0 0 0 0 0 0 0 1\\
0 0 0 0 1 0 0 0 0 0 0 0 1\\
0 0 0 0 0 1 0 0 0 0 0 0 1\\
0 0 0 0 0 0 1 0 0 0 0 0 1\\
0 0 0 0 0 0 0 1 0 0 0 0 1\\
0 0 0 0 0 0 0 0 1 0 0 0 1\\
0 0 0 0 0 0 0 0 0 1 0 0 1\\
0 0 0 0 0 0 0 0 0 0 1 0 1\\
0 0 0 0 0 0 0 0 0 0 0 1 1\\
\end{bNiceMatrix}
\]

\setlength{\tabcolsep}{2pt}

\begin{table}[h]
{\begin{center}
{
\begin{tabular}{|c||c|c|c|c|c|c|c|c|c|c|c|c|}
\hline
$n/k$ & 1 & 2 & 3 & 4 & 5 & 6 & 7 & 8 & 9 & 10 & 11 & 12\tabularnewline
\hline
\hline
1 & 0 &  &  &  &  &  &  &  &  &  &  & \tabularnewline
\hline
2 & 2 & 0 &  &  &  &  &  &  &  &  &  & \tabularnewline
\hline
3 & 2 & 1 & 0 &  &  &  &  &  &  &  &  & \tabularnewline
\hline
4 & 4 & 1 & 2 & 0 &  &  &  &  &  &  &  & \tabularnewline
\hline
5 & 4 & 3 & 2 & 1 & 0 &  &  &  &  &  &  & \tabularnewline
\hline
6 & 6 & 3 & 2 & 1 & 2 & 0 &  &  &  &  &  & \tabularnewline
\hline
7 & 6 & 4 & 3 & 2 & 2 & 1 & 0 &  &  &  &  & \tabularnewline
\hline
8 & 8 & 4 & 4 & 3 & 2 & 1 & 2 & 0 &  &  &  & \tabularnewline
\hline
9 & 8 & 5 & 4 & 3 & 3 & 2 & 2 & 1 & 0 &  &  & \tabularnewline
\hline
10 & 10 & 5 & 5 & 4 &  4&  3&  2&  1&  2& 0 &  & \tabularnewline
\hline
11 & 10 & 7 & 6 &  5& $4$  & 3 & 3 & 2 & 2 &   1 & 0 & \tabularnewline
\hline
12 & $12^o$ & 7 & $6^o$ & 5 & $4^o$ & $4^{I,o}$ & $4^{III,o}$ & $3^{III,o}$ & $2^{I,o}$ & 1 & $2^o$ & 0\tabularnewline
\hline
13 & 12 &$8^o$  &$6^{I}$  &$6^{I,o}$  & $5^{I,o}$  & $4^{I,o}$ & $4^{I,o}$  & $3^{I}$  & $2^{I}$ & $2^{I,o}$ & $2^{I,o}$  & 1 \tabularnewline
\hline
\end{tabular}}
\caption {Each cell refers to the highest minimum distance $d(n,k)$ for $n \leq13$ when $h=1$, and
examples of corresponding generator matrices $G_{12, k, d}^1~(6 \le k \le 11)$ and $G_{13, k, d}^1~ (3 \le k \le 11)$ }

\end{center}
}

\end{table}

\setlength{\tabcolsep}{1pt}

\begin{table}[h]
{\begin{center}
{
\begin{tabular}{|c||c|c|c|c|c|c|c|c|c|c|c|c|}
\hline
$n/k$ & 0 & 1 & 2 & 3 & 4 & 5 & 6 & 7 & 8 & 9 & 10 & 11 \tabularnewline
\hline
\hline
2 & $(2;0)$ &  &  &  &  &  &  &  &  &  &  & \tabularnewline
\hline
3 & $(2;1)$ & $(1;0)$ &  &  &  &  &  &  &  &  &  & \tabularnewline
\hline
4 & $(4;2)$ & $(1;1)$ & $(2;0)$ &  &  &  &  &  &  &  &  & \tabularnewline
\hline
5 & $(4;3)$ & $(3;2)$ & $(2;1)$ & $(1;0)$ &  &  &  &  &  &  &  & \tabularnewline
\hline
6 & $(6;4)$ & $(3;3)$ & $(2;2)$ & $(1;1)$ & $(2;0)$ &  &  &  &  &  &  & \tabularnewline
\hline
7 & $(6;5)$ & $(4;4)$ & $(3;3)$ & $(2;2)$ & $(2;1)$ & $(1;0)$ &  &  &  &  &  & \tabularnewline
\hline
8 & $(8;6)$ & $(4;5)$ & $(4;4)$ & $(3;3)$ & $(2;2)$ & $(1;1)$ & $(2;0)$ & &  &  &  & \tabularnewline
\hline
9 & $(8;7)$ & $(5;6)$ & $(4;5)$ & $(3;4)$ & $(3;3)$ & $(2;2)$ & $(2;1)$ & $(1;0)$ &  &  &  & \tabularnewline
\hline
10 & $(10;8)$ & $(5;7)$ & $(5;6)$ & $(4;5)$ &  $(4;4)$&  $(3;3)$&  $(2;2)$&  $(1;1)$ &  $(2;0)$ &  &  & \tabularnewline
\hline
11 & $(10;9)$ & $(7;8)$ & $(6;7)$ &  $(5;6)$& $(4;5)$  & $(3;4)$ & $(3;3)$ & $(2;2)$ & $(2;1)$ &   $(1;0)$ &  & \tabularnewline
\hline
12 & $(12;10)$ & $(7;9)$ & $(6;8)$ & $(5;7)$ & $(4;6)$ & $(4;5)$ & $(4;4)$ & $(3;3)$ & $(2;2)$ & $(1;1)$ & $(2;0)$ & \tabularnewline
\hline
13 & $(12;11)$ & $(8;10)$ & $(6;9)$  & $(6;8)$  & $(5;7)$  & $(4;6)$ & $(4;5)$  & $(3;4)$  & $(2;3)$ & $(2;2)$ & $(2;1)$  & $(1;0)$ \tabularnewline
\hline
\end{tabular}}
\caption {$[[n, k, d; c]]_2$  EAQECC with $(d;c)$  for $n \leq 13$ when $h=1$ based on Proposition 1 and Table 1}

\end{center}
}
\end{table}

\begin{example}
{\rm
Fix the hull dimension $h=2$.
For any $n$ with $k$ such that $1 \le k \le n \le 11$ and $n=12$ with $k~(1 \le k \le 4)$,  we ran exhaustive search to get optimal or $h_2$-optimal codes.

For $n=12$ with $k \ge 5$, we apply Constructions I, II or III to LCD codes or linear codes with $h=1$ of length $10$ and dimension $k-1$. More precisely, we construct optimal $[12, 5, 4]$, $[12, 6, 4]$, $[12, 8, 3]$, $[12, 9, 2]$, $[12, 10, 2]$ codes from $[10, 4, 4]$, $[10,5,3]$, $[10, 7, 2]$, $[10, 8, 2]$, $[10, 9, 1]$ codes with $h=0$ respectively by Construction III. On the other hand, we also construct a $[12, 7,3]$ code from a $[10,6,3]$ code with $h=0$ by Construction III. By exhaustive search, we check that it is $h_2$-optimal.

Let $n=13$. We construct
optimal $[13, 3, 7]$, $[13, 4, 6]$, $[13, 5,5]$, $[13, 6,4]$, $[13, 8, 4]$ codes by Construction III from LCD codes of length $11$ and dimensions $k=2, 3, 4, 5, 7$ respectively. We also construct an optimal $[13, 7, 4]$ code from a linear $[11, 6, 3]$ code with $h=1$ by Construction I. For $k=2, 10, 11$, it is easy to construct directly optimal or $h_2$-optimal $[13, k]$ codes. For $k=9$, it is known that there exist an optimal $[13, 9, 3]$ code with $h=2$ by Magma database.
}
\end{example}

$\bullet$ $n=12$ with $h=2$

\[
G_{12, 3,5}^{2}=
\begin{bNiceMatrix}
1 0 0 0 1 1 0 1 1 0 1 0\\
0 1 0 1 0 0 0 1 1 1 0 0\\
0 0 1 1 1 1 1 0 0 0 0 0\\
\end{bNiceMatrix},~
G_{12, 4,6}^{2}=
\begin{bNiceMatrix}
1 0 0 0 1 0 1 0 1 0 1 1\\
0 1 0 0 1 1 0 1 0 1 0 1\\
0 0 1 0 1 1 1 0 0 1 1 0\\
0 0 0 1 1 1 1 1 1 0 0 0\\
\end{bNiceMatrix}
\]

\[
G_{12, 5,4}^{2}=
\begin{bNiceMatrix}
1 1 0 0 0 1 1 0 0 0 0 0\\
1 0 1 0 0 0 1 1 1 0 0 1\\
0 0 0 1 0 0 0 1 1 1 0 1\\
1 0 0 0 1 0 1 0 1 0 1 1\\
1 0 0 0 0 1 0 0 1 1 1 1\\
\end{bNiceMatrix}, ~
G_{12, 6,4}^{2}=
\begin{bNiceMatrix}
1 1 1 1 1 1 0 0 1 1 0 0\\
1 0 1 0 0 0 1 0 0 0 1 0\\
1 0 0 1 0 0 1 0 0 0 0 1\\
1 0 0 0 1 0 0 0 0 0 1 1\\
0 0 0 0 0 1 1 0 0 1 0 1\\
0 0 0 0 0 0 0 0 1 1 1 1\\
\end{bNiceMatrix}
\]

\[
G_{12, 7,3}^{2}=
\begin{bNiceMatrix}
1 1 1 1 0 1 1 0 0 0 0 0\\
1 0 1 0 0 0 0 0 1 0 1 1\\
1 0 0 1 0 0 0 0 1 1 0 1\\
0 0 0 0 1 0 0 0 0 1 1 0\\
1 0 0 0 0 1 0 0 0 1 0 1\\
1 0 0 0 0 0 1 0 1 1 1 1\\
0 0 0 0 0 0 0 1 1 1 1 0\\
\end{bNiceMatrix},~
G_{12, 8,3}^{2}=
\begin{bNiceMatrix}
1 1 1 1 1 1 0 0 0 1 1 0\\
0 0 1 0 0 0 0 0 0 0 1 1\\
0 0 0 1 0 0 0 0 0 1 0 1\\
1 0 0 0 1 0 0 0 0 1 1 0\\
1 0 0 0 0 1 0 0 0 1 1 1\\
1 0 0 0 0 0 1 0 0 0 1 1\\
1 0 0 0 0 0 0 1 0 1 0 1\\
0 0 0 0 0 0 0 0 1 1 1 1\\
\end{bNiceMatrix}
\]

\[
G_{12, 9,2}^{2}=
\begin{bNiceMatrix}
1 1 0 1 1 0 0 0 0 0 0 0\\
0 0 1 0 0 0 0 0 0 0 1 1\\
1 0 0 1 0 0 0 0 0 0 1 1\\
1 0 0 0 1 0 0 0 0 0 1 0\\
0 0 0 0 0 1 0 0 0 0 1 1\\
0 0 0 0 0 0 1 0 0 0 1 1\\
0 0 0 0 0 0 0 1 0 0 1 1\\
0 0 0 0 0 0 0 0 1 0 1 1\\
0 0 0 0 0 0 0 0 0 1 1 1\\
\end{bNiceMatrix},~
G_{12, 10,2}^{2}=
\begin{bNiceMatrix}
1 1 1 1 1 1 1 1 1 1 1 1\\
1 0 1 0 0 0 0 0 0 0 0 0\\
1 0 0 0 1 0 0 0 0 0 0 0\\
1 0 0 0 0 1 0 0 0 0 0 0\\
1 0 0 0 0 0 1 0 0 0 0 0\\
1 0 0 0 0 0 0 1 0 0 0 0\\
1 0 0 0 0 0 0 0 1 0 0 0\\
1 0 0 0 0 0 0 0 0 1 0 0\\
1 0 0 0 0 0 0 0 0 0 1 0\\
1 0 0 0 0 0 0 0 0 0 0 1\\
\end{bNiceMatrix}
\]

$\bullet$ $n=13$ with $h=2$

\[
G_{13, 3,7}^{2}=
\begin{bNiceMatrix}
1 1 1 1 0 0 0 1 1 0 1 1 0\\
0 0 1 1 1 0 0 0 1 1 1 0 1\\
1 0 0 0 0 1 1 1 1 1 1 0 1\\
\end{bNiceMatrix},~
G_{13, 4,6}^{2}=
\begin{bNiceMatrix}
1 1 0 1 1 1 0 1 1 1 0 0 0\\
1 0 1 0 0 1 1 0 1 1 0 1 1\\
0 0 0 1 0 1 1 1 0 1 0 0 1\\
0 0 0 0 1 1 1 1 1 0 0 1 0\\
\end{bNiceMatrix}
\]

\[
G_{13, 5,5}^{2}=
\begin{bNiceMatrix}
1 1 0 1 1 1 1 0 1 1 0 0 0\\
1 0 1 0 0 0 1 1 0 0 1 0 0\\
1 0 0 1 0 0 1 0 1 0 0 1 0\\
1 0 0 0 1 0 1 0 1 0 1 0 1\\
0 0 0 0 0 1 1 1 0 0 0 1 1\\
\end{bNiceMatrix},~
G_{13, 6,4}^{2}=
\begin{bNiceMatrix}
1 1 1 1 0 0 0 0 0 0 0 0 0\\
1 0 1 0 0 0 0 0 1 0 1 1 1\\
1 0 0 1 0 0 0 1 0 1 0 1 0\\
0 0 0 0 1 0 0 1 1 0 0 1 1\\
0 0 0 0 0 1 0 0 0 1 1 1 0\\
0 0 0 0 0 0 1 1 1 0 1 0 1\\
\end{bNiceMatrix}
\]

\[
G_{13, 7,4}^{2}=
\begin{bNiceMatrix}
1 0 1 1 0 0 0 1 1 1 0 1 1\\
1 1 1 0 0 0 0 0 0 1 0 1 0\\
1 1 0 1 0 0 0 0 1 0 0 1 0\\
0 0 0 0 1 0 0 0 1 1 1 0 0\\
1 1 0 0 0 1 0 0 0 1 1 0 0\\
1 1 0 0 0 0 1 0 1 0 1 0 0\\
1 1 0 0 0 0 0 1 1 1 0 0 0\\
\end{bNiceMatrix},~
G_{13, 8,4}^{2}=
\begin{bNiceMatrix}
1 1 1 1 1 0 1 0 0 0 1 1 0\\
0 0 1 0 0 0 0 0 0 1 0 1 1\\
0 0 0 1 0 0 0 0 0 1 1 0 1\\
1 0 0 0 1 0 0 0 0 0 1 1 0\\
0 0 0 0 0 1 0 0 0 1 1 1 0\\
1 0 0 0 0 0 1 0 0 1 1 1 1\\
1 0 0 0 0 0 0 1 0 0 1 0 1\\
1 0 0 0 0 0 0 0 1 0 0 1 1\\
\end{bNiceMatrix}
\]

\[
G_{13, 9,3}^{2}=
\begin{bNiceMatrix}
1 0 0 0 0 0 0 0 0 0 0 1 1\\
0 1 0 0 0 0 0 0 1 0 1 0 0\\
0 0 1 0 0 0 0 0 1 0 1 1 1\\
0 0 0 1 0 0 0 0 1 0 0 0 1\\
0 0 0 0 1 0 0 0 1 0 0 1 0\\
0 0 0 0 0 1 0 0 0 0 1 0 1\\
0 0 0 0 0 0 1 0 0 0 1 1 0\\
0 0 0 0 0 0 0 1 1 0 0 1 1\\
0 0 0 0 0 0 0 0 0 1 1 1 1\\
\end{bNiceMatrix},~
G_{13, 10,2}^{2}=
\begin{bNiceMatrix}
1 0 0 0 0 0 0 0 0 0 1 0 0\\
0 1 0 0 0 0 0 0 0 0 0 1 0\\
0 0 1 0 0 0 0 0 0 0 0 0 1\\
0 0 0 1 0 0 0 0 0 0 0 0 1\\
0 0 0 0 1 0 0 0 0 0 0 0 1\\
0 0 0 0 0 1 0 0 0 0 0 0 1\\
0 0 0 0 0 0 1 0 0 0 0 0 1\\
0 0 0 0 0 0 0 1 0 0 0 0 1\\
0 0 0 0 0 0 0 0 1 0 0 0 1\\
0 0 0 0 0 0 0 0 0 1 0 0 1\\
\end{bNiceMatrix}
\]

\setlength{\tabcolsep}{1pt}

\begin{table}[h]
{\begin{center}
{

\begin{tabular}{|c||c|c|c|c|c|c|c|c|c|c|c|}
\hline
$n/k$ & 2 & 3 & 4 & 5 & 6 & 7 & 8 & 9 & 10 & 11 \tabularnewline
\hline
\hline
2 &  0 &  &  &  &  &  &  &  &  &   \tabularnewline
\hline
3 &  0 & 0 &  &  &  &  &  &  &  &   \tabularnewline
\hline
4 &  2 & 0 & 0 &  &  &  &  &  &  &   \tabularnewline
\hline
5 &  2 & 1 & 0 & 0 &  &  &  &  &  &   \tabularnewline
\hline
6 &  4 & 3 & 2 & 0 & 0 &  &  &  &  &  \tabularnewline
\hline
7 &  4 & 3 & 2 & 1 & 0 & 0 &  &  &  &   \tabularnewline
\hline
8 &  4 & 4 & 3 & 2 & 2 & 0 & 0 &  &  &  \tabularnewline
\hline
9 &  4 & 4 & 4 & 3 & 2 & 1 & 0 & 0 &  &   \tabularnewline
\hline
10 &  6 & 4 & 4 &  3&  3&  2&  2&  0& 0 &  \tabularnewline
\hline
11 &  6 & 5 &  4&  4& 4 & 3 & 2 & 1 & 0 & 0 \tabularnewline
\hline
12 & $8^o$ & 5 & $6^o$ & $4^{III,o}$ & $4^{III,o}$ & $3^{III}$ & $3^{III,o}$ & $2^{III,o}$ & $2^{III,o}$ & 0  \tabularnewline
\hline
13 & $8^{o}$ & $7^{III,o}$ & $6^{III,o}$ & $5^{III,o}$  & $4^{III,o}$  & $4^{I,h_1,o}$  & $4^{III,o}$ & $3^o$  & $2^o$ & 1  \tabularnewline
\hline
\end{tabular}}
\caption {Each cell refers to the highest minimum distance $d(n,k)$ for $n \leq 13$ when $h=2$, and
examples of corresponding generator matrices $G_{12, k, d}^2~(3 \le k \le 10)$ and $G_{13, k, d}^2~ (3 \le k \le 10)$
}
\end{center}
}
\end{table}

\setlength{\tabcolsep}{1pt}

\begin{table}[h]
{\begin{center}
{

\begin{tabular}{|c||c|c|c|c|c|c|c|c|c|c|c|}
\hline
$n/k$ & 0 & 1 & 2 & 3 & 4 & 5 & 6 & 7 & 8 & 9 \tabularnewline
\hline
\hline
4 &  (2;0) &  &  &  &  &  &  &  &  &   \tabularnewline
\hline
5 &  (2;1) & (1;0) &  &  &  &  &  &  &  &   \tabularnewline
\hline
6 &  (4;2) & (3;1) & (2;0) &  &  &  &  &  &  &  \tabularnewline
\hline
7 &  (4;3) & (3;2) & (2;1) & (1;0) &  &  &  &  &  &   \tabularnewline
\hline
8 &  (4;4) & (4;3) & (3;2) & (2;1) & (2;0) &  &  &  &  &  \tabularnewline
\hline
9 &  (4;5) & (4;4) & (4;3) & (3;2) & (2;1) & (1;0) &  &  &  &   \tabularnewline
\hline
10 &  (6;6) & (4;5) & (4;4) &  (3;3)&  (3;2)&  (2;1)&  (2;0)&  0& 0 &  \tabularnewline
\hline
11 &  (6;7) & (5;6) &  (4;5)&  (4;4)& (4;3) & (3;2) & (2;1) & (1;0) &  &  \tabularnewline
\hline
12 & $(8;8)$ & (5;7) & $(6;6)$ & $(4;5)$ & $(4;4)$ & $(3;3)$ & $(3;2)$ & $(2;1)$ & $(2;0)$ &   \tabularnewline
\hline
13 & $(8;9)$ & $(7;8)$ & $(6;7)$ & $(5;6)$  & $(4;5)$  & $(4;4)$  & $(4;3)$ & $(3;2)$  & $(2;1)$ & (1;0)  \tabularnewline
\hline
\end{tabular}}
\caption {$[[n, k, d; c]]_2$  EAQECC with $(d;c)$  for $n \leq 13$ when $h=2$ based on Proposition 1 and
Table 3}

\end{center}
}
\end{table}

\begin{example}
{\rm
Fix the hull dimension $h=3$.

Since $h=3$, the code length $n$ should be at least $6$. If $n-k \le 2$, then there does not exist an $[n, k]$ code with $h=3$. If $k=3$, we use the optimal minimum distances of self-orthogonal $[n,3]$ codes from~\cite{Bou}.

For any $n$ with $k$ such that $3 \le k \le n \le 11$ and $n=12$ with $k=3, 4$,  we ran exhaustive search to obtain optimal or $h_3$-optimal codes.

Using Construction III, we construct an optimal $[12, 5, 4]$ code with $h=3$ from a $[10, 4, 4]$ code, and an $h_3$-optimal $[12, 6, 4]$ code from a $[10, 5, 3]$ code with $h=2$. We further construct an optimal $[12, 7,4]$ code from a $[10, 6, 3]$ code with $h=1$ by Construction III.

  We construct $h_3$-optimal $[12, 8, 2]$ and $[12, 9, 2]$ codes from $[10, 7, 2]$ and $[10, 8,2]$ codes with $h=2$, respectively by Construction I. This is justified by exhaustive search that there are no $[12, 6, 4]$, $[12, 8, 3]$ codes.

For $n=13$, we construct $[13, 4,4]$, $[13, 5,4]$, $[13, 6,4]$, $[13, 7, 3]$, $[13, 8, 2]$, $[13, 9, 2]$ codes with $h=3$ from $[11, k]$ codes with $h=3$ $(3 \le k \le 8)$ by Construction IV. Similarly we construct
$[13, 4,5]$, $[13, 5,4]$, $[13, 6,4]$, $[13, 7, 4]$, $[13, 8, 3]$, $[13, 10, 2]$ codes with $h=3$ from $[11, k]$ codes with $h=2$ $(3 \le k \le 7, k=9)$ by Construction I.

}
\end{example}

$\bullet$ $n=12$ with $h=3$

\[
G_{12, 4,4}^{3}=
\begin{bNiceMatrix}
1 0 0 0 0 1 1 1 1 0 0 0\\
0 1 0 0 1 0 1 1 0 0 0 0\\
0 0 1 0 1 1 0 1 0 0 0 0\\
0 0 0 1 1 1 1 0 0 0 0 0\\
\end{bNiceMatrix},~
G_{12, 5,4}^{3}=
 \begin{bNiceMatrix}
1 0 0 1 0 0 1 0 1 1 0 0\\
0 1 0 1 0 0 0 1 1 1 1 0\\
0 0 1 1 0 0 1 1 0 0 1 0\\
0 0 0 0 1 0 1 1 0 1 0 0\\
0 0 0 0 0 1 1 1 1 0 0 0\\
\end{bNiceMatrix}
\]

\[
G_{12, 6,4}^{3}=
\begin{bNiceMatrix}
1 1 1 1 0 0 0 0 0 0 0 0\\
1 0 1 0 0 0 0 1 1 0 0 1\\
1 0 0 1 0 0 0 0 1 1 1 0\\
0 0 0 0 1 0 0 1 0 1 1 0\\
0 0 0 0 0 1 0 1 1 0 1 0\\
0 0 0 0 0 0 1 1 1 1 0 0\\
\end{bNiceMatrix},~
G_{12, 7,4}^{3}=
\begin{bNiceMatrix}
1 1 1 1 1 1 1 1 1 1 1 1\\
1 0 1 0 0 0 0 0 0 1 0 1\\
1 0 0 1 0 0 0 0 1 0 0 1\\
0 0 0 0 1 0 0 0 1 1 1 0\\
1 0 0 0 0 1 0 0 0 1 1 0\\
1 0 0 0 0 0 1 0 1 0 1 0\\
1 0 0 0 0 0 0 1 1 1 0 0\\
\end{bNiceMatrix}
\]

\[
G_{12, 8,2}^{3}=
\begin{bNiceMatrix}
1 1 1 1 0 0 0 0 0 0 0 0\\
1 0 1 0 0 0 0 0 0 0 0 1\\
1 0 0 1 0 0 0 0 0 0 1 0\\
0 0 0 0 1 0 0 0 0 1 0 0\\
0 0 0 0 0 1 0 0 0 1 0 0\\
0 0 0 0 0 0 1 0 0 1 0 0\\
0 0 0 0 0 0 0 1 0 1 0 0\\
0 0 0 0 0 0 0 0 1 1 0 0\\
\end{bNiceMatrix},~
G_{12, 9,2}^{3}=
\begin{bNiceMatrix}
1 1 0 0 0 0 0 0 0 0 0 0\\
0 0 1 0 0 0 0 0 0 0 0 1\\
0 0 0 1 0 0 0 0 0 0 1 0\\
0 0 0 0 1 0 0 0 0 0 1 0\\
0 0 0 0 0 1 0 0 0 0 1 0\\
0 0 0 0 0 0 1 0 0 0 1 0\\
0 0 0 0 0 0 0 1 0 0 1 0\\
0 0 0 0 0 0 0 0 1 0 1 0\\
0 0 0 0 0 0 0 0 0 1 1 0\\
\end{bNiceMatrix}
\]

$\bullet$ $n=13$ with $h=3$

\[
G_{13, 4,5}^{3}=
\begin{bNiceMatrix}
1 0 1 1 1 1 1 0 0 0 0 0 0\\
0 0 1 0 0 0 1 1 0 1 1 0 1\\
0 0 0 1 0 1 0 0 0 1 1 1 0\\
1 1 0 0 1 1 1 1 1 0 0 0 0\\
\end{bNiceMatrix},~
G_{13, 5,4}^{3}=
\begin{bNiceMatrix}
1 0 1 1 1 0 0 0 0 0 0 0 0\\
1 1 1 0 0 0 1 1 0 0 1 0 0\\
1 1 0 1 0 0 1 0 1 1 0 0 0\\
1 1 0 0 1 0 1 1 0 1 0 0 0\\
0 0 0 0 0 1 1 1 1 0 0 0 0\\
\end{bNiceMatrix}
\]
\[
G_{13, 6,4}^{3}=
\begin{bNiceMatrix}
1 0 0 1 1 1 0 0 0 0 0 0 0\\
0 0 1 0 0 0 0 0 1 1 1 0 1\\
1 1 0 1 0 0 0 1 1 0 0 1 0\\
1 1 0 0 1 0 0 1 0 1 1 0 0\\
1 1 0 0 0 1 0 1 1 0 1 0 0\\
0 0 0 0 0 0 1 1 1 1 0 0 0\\
\end{bNiceMatrix},~
G_{13, 7,4}^{3}=
\begin{bNiceMatrix}
1 0 1 0 1 1 0 0 0 0 0 0 0\\
1 1 1 0 0 0 0 0 1 0 1 0 1\\
0 0 0 1 0 0 0 0 1 1 0 0 1\\
1 1 0 0 1 0 0 0 0 1 1 1 0\\
1 1 0 0 0 1 0 0 1 0 1 1 0\\
0 0 0 0 0 0 1 0 1 1 0 1 0\\
0 0 0 0 0 0 0 1 1 1 1 0 0\\
\end{bNiceMatrix}
\]
\[
G_{13, 8,3}^{3}=
\begin{bNiceMatrix}
1 0 1 0 0 0 0 1 1 0 0 0 0\\
1 1 1 0 0 0 0 0 0 1 1 0 1\\
0 0 0 1 0 0 0 0 0 0 1 0 1\\
0 0 0 0 1 0 0 0 0 1 0 0 1\\
0 0 0 0 0 1 0 0 0 1 1 1 0\\
0 0 0 0 0 0 1 0 0 0 1 1 0\\
1 1 0 0 0 0 0 1 0 1 0 1 0\\
1 1 0 0 0 0 0 0 1 1 1 0 0\\
\end{bNiceMatrix},~
G_{13, 9,2}^{3}=
\begin{bNiceMatrix}
1 0 1 1 0 0 0 0 0 0 0 0 0\\
1 1 1 0 0 0 0 0 0 0 1 1 1\\
1 1 0 1 0 0 0 0 0 0 0 1 1\\
0 0 0 0 1 0 0 0 0 0 1 0 1\\
0 0 0 0 0 1 0 0 0 0 1 1 0\\
0 0 0 0 0 0 1 0 0 0 1 0 0\\
0 0 0 0 0 0 0 1 0 0 1 0 0\\
0 0 0 0 0 0 0 0 1 0 1 0 0\\
0 0 0 0 0 0 0 0 0 1 1 0 0\\
\end{bNiceMatrix}
\]

\[
G_{13, 10,2}^{3}=
\begin{bNiceMatrix}
1 0 1 1 1 1 1 1 1 1 1 0 0\\
1 1 1 0 0 0 0 0 0 0 0 0 1\\
1 1 0 1 0 0 0 0 0 0 0 1 0\\
1 1 0 0 1 0 0 0 0 0 0 0 0\\
1 1 0 0 0 1 0 0 0 0 0 0 0\\
1 1 0 0 0 0 1 0 0 0 0 0 0\\
1 1 0 0 0 0 0 1 0 0 0 0 0\\
1 1 0 0 0 0 0 0 1 0 0 0 0\\
1 1 0 0 0 0 0 0 0 1 0 0 0\\
1 1 0 0 0 0 0 0 0 0 1 0 0\\
\end{bNiceMatrix}
\]

\setlength{\tabcolsep}{1pt}

\begin{table}[h]
{\begin{center}
{

\begin{tabular}{|c||c|c|c|c|c|c|c|c|}
\hline
$n/k$ & 3 & 4 & 5 & 6 & 7 & 8 & 9 & 10  \tabularnewline
\hline
\hline
6 & 2 & 0 & 0 & 0 &  &  &  &    \tabularnewline
\hline
7 & 4 & 3 & 0 & 0 & 0  &  &  &  \tabularnewline
\hline
8 & 4 & 3 & 2 & 0 & 0 & 0 &  & \tabularnewline
\hline
9 & 4 & 4 & 3 & 2 & 0 & 0 &  & \tabularnewline
\hline
10 & 4 & 4 & 4 & 2 &  2&  0&  0  & \tabularnewline
\hline
11 & 4 & 4 & 4 &  3& 2 & 2 & 0  &  \tabularnewline
\hline
12 & $6^o$ & 4 & $4^{III,o}$ & $4^{III,o}$  & $4^{III, h_1,o}$ & $2^{III}$ & $2^{II,o}$ & \tabularnewline
\hline
13 & 6 & $\ge 5^{I}$ & $\ge 4^{I,IV}$ & $4^{I,IV, o}$ & $4^{I,o}$ & $\ge 3^{I}$ & $\ge 2^{IV}$ & $2^{I,o}$ \tabularnewline
\hline
\end{tabular}}
\caption {Each cell refers to the highest minimum distance $d(n,k)$ for $n \leq13$ when $h=3$, and
examples of corresponding generator matrices $G_{12, k, d}^3~(4 \le k \le 9)$ and $G_{13, k, d}^3~ (4 \le k \le 10)$}
\end{center}
}
\end{table}

\setlength{\tabcolsep}{1pt}

\begin{table}[h]
{\begin{center}
{

\begin{tabular}{|c||c|c|c|c|c|c|c|c|}
\hline
$n/k$ & 0 & 1 & 2 & 3 & 4 & 5 & 6 & 7  \tabularnewline
\hline
\hline
6 & (2;0) &  &  &  &  &  &  &    \tabularnewline
\hline
7 & (4;1) & (3;0) &  &  &   &  &  &  \tabularnewline
\hline
8 & (4;2) & (3;1) & (2;0) &  &  &  &  & \tabularnewline
\hline
9 & (4;3) & (4;2) & (3;1) & (2;0) &  &  &  & \tabularnewline
\hline
10 & (4;4) & (4;3) & (4;2) & (2;1) &  (2;0)&  &    & \tabularnewline
\hline
11 & (4;5) & (4;4) & (4;3) &  (3;2)& (2;1) & (2;0) &   &  \tabularnewline
\hline
12 & $(6;6)$ & (4;5) & $(4;4)$ & $(4;3)$  & $(4;2)$ & $(2;1)$ & $(2;0)$ & \tabularnewline
\hline
13 & (6;7) & $(\ge 5;6)$ & $(\ge 4;5)$ & $(4;4)$ & $(4;3)$ & $(\ge 3;2)$ & $(\ge 2;1)$ & $(2;0)$ \tabularnewline
\hline
\end{tabular}}
\caption {$[[n, k, d; c]]_2$  EAQECC with $(d;c)$  for $n \leq 13$ when $h=3$ based on Proposition 1 and
Table 5}
\end{center}
}
\end{table}

\begin{example}
{\rm
Fix the hull dimension $h=4$.

Since $h=4$, the code length $n$ should be at least $8$. If $n-k \le 3$, then there does not exist a $[n, k]$ code with $h=4$. If $k=4$, we use the optimal minimum distances of self-orthogonal $[n,4]$ codes from~\cite{Bou}.

For any $n$ with $k$ such that $4 \le k \le n \le 11$ and $n=12$ with $k=4$,  we ran exhaustive search to obtain optimal or $h_4$-optimal codes.

We construct optimal $[12, 5, 4]$ and $[12, 6,4]$ codes with $h=4$ from $[10, 4, 4]$ and  $[10, 5, 3]$ codes with $h=3$ respectively by Construction I.

We construct $h_4$-optimal $[12, 7, 3]$, $[12, 8, 2]$ codes with $h=4$ from $[10, 6, 1]$, $[10, 7, 2]$ codes respectively with $h=3$ by Construction I. This is justified by exhaustive checking that
there are no $[12, 7, 4]$ and $[12, 8, 3]$ codes with $h=4$.

For $n=13$, we obtain $[13, 5,4]$, $[13, 6,4]$, $[13, 7, 3]$, $[13, 8, 2]$ codes with $h=4$ from $[11, k]$ codes with $h=4$ $(4 \le k \le 7)$ by Construction IV.
Similarly, we construct $[13, 4,4]$, $[13, 5,4]$, $[13, 6, 4]$, $[13, 7, 3]$, $[13, 8,2]$, $[13, 9,2]$ codes with $h=4$ from $[11, k]$ codes with $h=3$ $(3 \le k \le 8)$ by Construction I.

}
\end{example}

$\bullet$ $n=12$ with $h=4$

\[
G_{12, 5,4}^{4}=
\begin{bNiceMatrix}
1 0 1 0 0 0 0 0 0 0 1 1\\
0 0 1 0 0 0 0 1 1 1 1 0\\
0 0 0 1 0 0 1 0 1 1 0 0\\
0 0 0 0 1 0 1 1 0 1 0 0\\
0 0 0 0 0 1 1 1 1 0 0 0\\
\end{bNiceMatrix},~
G_{12, 6,4}^{4}=
\begin{bNiceMatrix}
1 0 1 0 1 1 0 0 0 0 0 0\\
1 1 1 0 0 0 0 1 1 0 0 1\\
0 0 0 1 0 0 0 0 1 1 1 0\\
1 1 0 0 1 0 0 1 0 1 1 0\\
1 1 0 0 0 1 0 1 1 0 1 0\\
0 0 0 0 0 0 1 1 1 1 0 0\\
\end{bNiceMatrix}
\]
\[
G_{12, 7,3}^{4}=
\begin{bNiceMatrix}
1 0 0 0 0 1 1 1 1 0 0 1 \\
1 1 0 0 0 1 1 1 0 1 0 0 \\
0 0 1 0 0 1 1 1 0 1 0 1 \\
0 0 0 1 0 0 1 1 0 0 1 0 \\
0 0 0 0 1 1 1 0 0 0 1 0 \\
0 0 0 0 0 1 0 1 1 1 1 1 \\
0 0 0 0 0 0 1 0 0 0 1 1 \\
\end{bNiceMatrix},~
G_{12, 8,2}^{4}=
\begin{bNiceMatrix}
1 0 1 0 0 0 0 0 0 0 0 0\\
1 1 1 0 0 0 0 0 0 0 0 1\\
0 0 0 1 0 0 0 0 0 0 1 0\\
0 0 0 0 1 0 0 0 0 1 0 0\\
0 0 0 0 0 1 0 0 0 1 0 0\\
0 0 0 0 0 0 1 0 0 1 0 0\\
0 0 0 0 0 0 0 1 0 1 0 0\\
0 0 0 0 0 0 0 0 1 1 0 0\\
\end{bNiceMatrix}
\]

$\bullet$ $n=13$ with $h=4$

\[
G_{13, 5,4}^{4}=
\begin{bNiceMatrix}
0 0 1 0 0 0 0 1 1 1 1 0 0\\
0 0 0 1 0 0 1 0 1 1 0 0 0\\
0 0 0 0 1 0 1 1 0 1 0 0 0\\
0 0 0 0 0 1 1 1 1 0 0 0 0\\
\end{bNiceMatrix},~
G_{13, 6,4}^{4}=
\begin{bNiceMatrix}
1 0 1 0 1 1 0 0 0 0 0 0 0\\
1 1 1 0 0 0 0 1 1 0 0 1 0\\
0 0 0 1 0 0 0 0 1 1 1 0 0\\
1 1 0 0 1 0 0 1 0 1 1 0 0\\
1 1 0 0 0 1 0 1 1 0 1 0 0\\
0 0 0 0 0 0 1 1 1 1 0 0 0\\
\end{bNiceMatrix}
\]
\[
G_{13, 7,3}^{4}=
\begin{bNiceMatrix}
1 0 1 1 1 0 0 0 0 0 0 0 0\\
1 1 1 0 0 0 0 0 1 0 0 0 1\\
1 1 0 1 0 0 0 0 1 0 0 1 0\\
1 1 0 0 1 0 0 0 1 1 1 0 0\\
0 0 0 0 0 1 0 0 0 1 1 0 0\\
0 0 0 0 0 0 1 0 1 0 1 0 0\\
0 0 0 0 0 0 0 1 1 1 0 0 0\\
\end{bNiceMatrix},~
G_{13, 8,2}^{4}=
\begin{bNiceMatrix}
1 0 1 0 0 0 0 0 0 0 0 0 0\\
1 1 1 0 0 0 0 0 0 0 0 1 0\\
0 0 0 1 0 0 0 0 0 0 1 0 0\\
0 0 0 0 1 0 0 0 0 1 0 0 0\\
0 0 0 0 0 1 0 0 0 1 0 0 0\\
0 0 0 0 0 0 1 0 0 1 0 0 0\\
0 0 0 0 0 0 0 1 0 1 0 0 0\\
0 0 0 0 0 0 0 0 1 1 0 0 0\\
\end{bNiceMatrix}
\]

\[
G_{13, 9,2}^{4}=
\begin{bNiceMatrix}
1 0 1 0 0 0 0 0 0 0 0 0 0\\
1 1 1 0 0 0 0 0 0 0 1 1 1\\
0 0 0 1 0 0 0 0 0 0 0 1 1\\
0 0 0 0 1 0 0 0 0 0 1 0 1\\
0 0 0 0 0 1 0 0 0 0 1 1 0\\
0 0 0 0 0 0 1 0 0 0 1 0 0\\
0 0 0 0 0 0 0 1 0 0 1 0 0\\
0 0 0 0 0 0 0 0 1 0 1 0 0\\
0 0 0 0 0 0 0 0 0 1 1 0 0\\
\end{bNiceMatrix}
\]

\setlength{\tabcolsep}{1pt}

\begin{table}[h]
{\begin{center}
{
\begin{tabular}{|c||c|c|c|c|c|c|}
\hline
$n/k$ & 4 & 5 & 6 & 7 & 8 & 9  \tabularnewline
\hline
\hline
8 & 4 & 0 & 0 & 0 & 0  & 0   \tabularnewline
\hline
9 & 4 & 2 & 0 & 0 & 0 & 0 \tabularnewline
\hline
10 & 4 & 4 & 2 & 0 &  0  & 0    \tabularnewline
\hline
11 & 4 & 4 & 3 & 2 & 0 & 0  \tabularnewline
\hline
12 & 4 & $4^{I,o}$ & $4^{I,o}$ & $3^I$ & $2^I$  & 0  \tabularnewline
\hline
13 & 4 & $\ge 4^{I, IV}$ & $4^{I, IV, o}$& $\ge 3^{I, IV}$ &
 $\ge 2^{I, IV}$ & $\ge 2^{I}$  \tabularnewline
 \hline
\end{tabular}
\caption {Each cell refers to the highest minimum distance $d(n,k)$  for $n \le 13$ when $h=4$, and
examples of corresponding generator matrices $G_{12, k, d}^4~(5 \le k \le 8)$ and $G_{13, k, d}^4~ (5 \le k \le 9)$}
}
\end{center}
}
\end{table}

\setlength{\tabcolsep}{1pt}

\begin{table}[h]
{\begin{center}
{
\begin{tabular}{|c||c|c|c|c|c|c|}
\hline
$n/k$ & 0 & 1 & 2 & 3 & 4 & 5  \tabularnewline
\hline
\hline
8 & (4;0) &  &  &  &   &    \tabularnewline
\hline
9 & (4;1) & (2;0) &  &  &  &  \tabularnewline
\hline
10 & (4;2) & (4;1) & (2;0) &  &    &     \tabularnewline
\hline
11 & (4;3) & (4;2) & (3;1) & (2;0) &  &   \tabularnewline
\hline
12 & (4;4) & $(4;3)$ & $(4;2)$ & $(3;1)$ & $(2;0)$  &   \tabularnewline
\hline
13 & (4;5) & $(\ge 4;4)$ & $(4;3)$& $(\ge 3;2)$ &
 $(\ge 2;1)$ & $(\ge 2;0)$  \tabularnewline
 \hline
\end{tabular}
\caption {$[[n, k, d; c]]_2$  EAQECC with $(d;c)$  for $n \leq 13$ when $h=4$ based on Proposition 1 and
Table 7}
}
\end{center}
}
\end{table}

\begin{example}
{\rm
Fix the hull dimension $h=5$.

Since $h=5$, the code length $n$ should be at least $10$. If $n-k \le 4$, then there does not exist a $[n, k]$ code with $h=5$. If $k=5$, we use the optimal minimum distances of self-orthogonal $[n,5]$ codes from~\cite{Bou}.

For $n=10, 11$ with $k=4,5$, we ran exhaustive search to obtain optimal or $h_5$-optimal codes.

It is well known that there is a self-orthogonal $[12, 5, 4]$ code~\cite{Bou}, which is optimal.
We construct $h_5$-optimal $[12, 6, 3]$ and $[12, 7, 3]$ codes from $[10, 5, 4]$ and $[10, 6, 2]$ codes with $h=4$ by Construction I. This is justified by exhaustive checking that there are no $[12, 6, 4]$ and $[12, 7, 4]$ codes with $h=5$.

For $n=13$, we obtain $[13, 6,3]$, $[13, 7,3]$ codes with $h=5$ from $[11, k]$ codes with $h=5$ $(5 \le k \le 6)$ by Construction IV. Similarly we construct $[13, 6, 4]$, $[13, 7, 3]$, $[13, 8, 2]$ codes from $[11, k]$ codes with $h=4$ $(5 \le k \le 7)$ by Construction I.

Although we cannot construct a $[13, 7, 4]$ code with $h=5$ from Constructions I and IV, we observe the following. Using the unique self-dual $[12, 6, 4]$ code $B_{12}$~\cite{Ple} with the below generator matrix $G_{12}$, we obtain an optimal $[13, 7, 4]$ code $C_{13,7,4}$ with the below generator matrix $G_{13,7,4}$ by augmenting a coset leader ${\bf v}= (000000010101)$ to $B_{12}$ because the covering radius of $B_{12}$ is 3. We show that $h(C_{13,7,4})=5$ in what follows. The top row ${\bf r}_1= (1~|~{\bf v})$ of $G_{13,7,4}$ is orthogonal to only five rows ${\bf r}_2, {\bf r}_3, {\bf r}_5, {\bf r}_6, {\bf r}_7$ of $G_{13}$. Therefore, the hull of $C_{13,7,4}$ consists of these five rows, resulting in $h(C_{13,7,4})=5$.

\[
G_{12}=
\begin{bNiceMatrix}
111100000000 \\
001111000000 \\
000011110000 \\
000000111100 \\
000000001111 \\
010101010101 \\
\end{bNiceMatrix}
,\qquad
G_{13,7,4}=
\left[\begin{array}{c|c}
1&000000010101 \\
\hline
0&111100000000 \\
0&001111000000 \\
0&000011110000 \\
0&000000111100 \\
0&000000001111 \\
0&010101010101 \\
\end{array}\right]
\]

$\bullet$ $n=12$ with $h=5$

\[
G_{12, 5,4}^{5}=
\begin{bNiceMatrix}
1 0 1 0 0 0 0 0 0 1 1 0\\
0 0 1 0 0 0 0 1 1 1 0 0\\
1 1 0 1 0 0 1 0 1 1 0 0\\
1 1 0 0 1 0 1 1 0 1 0 0\\
0 0 0 0 0 1 1 1 1 0 0 0\\
\end{bNiceMatrix},~
G_{12, 6,3}^{5}=
\begin{bNiceMatrix}
1 0 1 0 1 1 0 0 0 0 0 0\\
1 1 1 0 0 0 0 0 1 1 1 1\\
0 0 0 1 0 0 0 1 0 1 1 1\\
1 1 0 0 1 0 0 1 1 0 0 1\\
1 1 0 0 0 1 0 1 1 0 1 0\\
0 0 0 0 0 0 1 1 1 1 0 0\\
\end{bNiceMatrix}
\]
\[
G_{12, 7,3}^{5}=
\begin{bNiceMatrix}
1 0 0 0 0 1 0 1 1 1 0 1\\
1 1 0 0 0 0 0 0 0 0 1 1\\
0 0 1 0 0 1 1 1 0 0 0 0\\
0 0 0 1 0 0 0 0 0 1 1 1\\
0 0 0 0 1 0 1 1 1 0 0 0\\
0 0 0 0 0 1 0 1 1 0 1 1\\
0 0 0 0 0 0 1 1 0 0 1 1\\
\end{bNiceMatrix}
\]

$\bullet$ $n=13$ with $h=5$

\[
G_{13, 5,4}^{5}=
\begin{bNiceMatrix}
1 0 0 0 0 0 0 1 1 0 1 0 0\\
0 1 0 0 1 0 1 0 1 1 1 0 0\\
0 0 1 0 0 0 0 1 1 1 0 0 0\\
0 0 0 1 1 0 0 1 1 0 0 0 0\\
0 0 0 0 0 1 1 1 1 0 0 0 0\\
\end{bNiceMatrix},~
G_{13, 6,4}^{5}=
\begin{bNiceMatrix}
1 0 1 1 1 0 0 0 0 0 0 0 0\\
1 1 1 0 0 0 0 0 1 1 1 0 1\\
1 1 0 1 0 0 0 0 1 1 1 1 0\\
1 1 0 0 1 0 0 1 0 1 1 0 0\\
0 0 0 0 0 1 0 1 1 0 1 0 0\\
0 0 0 0 0 0 1 1 1 1 0 0 0\\
\end{bNiceMatrix}
\]
\[
G_{13, 7,4}^{5}=
\begin{bNiceMatrix}
1000000010101\\
0111100000000\\
0001111000000\\
0000011110000\\
0000000111100\\
0000000001111\\
0010101010101\\
\end{bNiceMatrix},~
G_{13, 8,2}^{5}=
\begin{bNiceMatrix}
1 0 1 0 0 0 0 0 0 0 0 0 0\\
1 1 1 0 0 0 0 0 0 0 1 1 1\\
0 0 0 1 0 0 0 0 0 0 0 1 1\\
0 0 0 0 1 0 0 0 0 0 1 0 1\\
0 0 0 0 0 1 0 0 0 0 1 1 0\\
0 0 0 0 0 0 1 0 0 1 0 0 0\\
0 0 0 0 0 0 0 1 0 1 0 0 0\\
0 0 0 0 0 0 0 0 1 1 0 0 0\\
\end{bNiceMatrix}
\]

}
\end{example}


\begin{table}[h]
{\begin{center}
{
\begin{tabular}{|c||c|c|c|c|}
\hline
$n/k$  & 5 & 6 & 7 & 8 \tabularnewline
\hline
\hline
10 & 2 & 0 & 0 & 0     \tabularnewline
\hline
11 & 4 & 3 & 0 & 0     \tabularnewline
\hline
12 & $4^o$ & $3^I$ & $3^I$ & 0 \tabularnewline
\hline
13 & 4 & $4^{I,o}$ & $4^o$ & $\ge 2^I$ \tabularnewline
\hline
\end{tabular}
\caption {Each cell refers to the highest minimum distance $d(n,k)$ for $n \leq 13$ when $h=5$, and
examples of corresponding generator matrices $G_{12, k, d}^5~(5 \le k \le 7)$ and $G_{13, k, d}^5~ (5 \le k \le 8)$
}
}
\end{center}
}
\end{table}
\medskip

\begin{table}[h]
{\begin{center}
{
\begin{tabular}{|c||c|c|c|c|}
\hline
$n/k$  & 0 & 1 & 2 & 3 \tabularnewline
\hline
\hline
10 & (2;0) &  &  &      \tabularnewline
\hline
11 & (4;1) & (3;0) &  &      \tabularnewline
\hline
12 & $(4;2)$ & $(3;1)$ & $(3;0)$ &  \tabularnewline
\hline
13 & (4;3) & $(4;2)$ & $(4;1)$ & $(\ge 2;0)$ \tabularnewline
\hline
\end{tabular}
\caption {$[[n, k, d; c]]_2$  EAQECC with $(d;c)$  for $n \leq 13$ when $h=5$ based on Proposition 1 and
Table 9}
}
\end{center}
}
\end{table}

There are not many known $[[n, k, d; c]]$ EAQECC when $n \le 13$.
We compare our results with some known EAQECC in Table 11. In fact, the parameters in boldface in the third column of the table are better than the currently known parameters from~\cite{NguKim},~\cite{SokQia}.

\renewcommand{\arraystretch}{2}

\begin{table}[h] \label{tab-EAQECC-sum}
{\begin{center}
{
\begin{tabular}{c|c|c|c}
\hline
\hline
currently known EAQECC & Ref & our related EAQECC & Tables\tabularnewline
\hline
$[[9,1,3;1]]_2$  & \cite{NguKim} & $[[{\bf 9,2,3;1}]]_2$ & Table 6 \tabularnewline
$[[12,1,7;9]]_2$ &\cite{SokQia} &  $[[12,1,7;9]]_2$ & Table 2 \tabularnewline
$[[12,3,5;7]]_2$ &\cite{SokQia} &  $[[12,3,5;7]]_2$ & Table 2 \tabularnewline
$[[12,4,4;6]]_2$ &\cite{SokQia} & $[[12,4,4;6]]_2$, $[[12,2,6;6]]_2$ &  Table 2, Table 4   \tabularnewline
$[[12,5,3;5]]_2$ &\cite{SokQia} & $[[{\bf 12,5,4;5}]]_2$,
$[[12,3,4;5]]_2$ &  Table 2, Table 4    \tabularnewline
$[[13,7,3;4]]_2$ &\cite{SokQia} &  $[[13,7,3;4]]_2$,
$[[13,5,4;4]]_2$ &  Table 2, Table 4    \tabularnewline
$[[13,3,5;8]]_2$ &\cite{SokQia} &  $[[{\bf 13,3,6;8}]]_2$,
$[[13,1,7;8]]_2$ &  Table 2, Table 4     \tabularnewline
\hline
\hline
\end{tabular}}
\caption {Comparison with some known EAQECC}
\end{center}
}
\end{table}

\section{Conclusion}
This paper has introduced a systematic and efficient method to construct binary optimal or possibly optimal $[n, k]$ codes of lengths up to 13 with respect to hull dimensions 1-5. These codes are used to construct EAQECC with the best known parameters.

The complexity of Constructions I-IV mainly depends on the binary vectors ${\bf x}$ of length $n$, whose cardinality is at most $2^{n-1}$ due to the parity of ${\bf x}$. This complexity can be reduced if we consider the standard generator matrix $G$ in Theorems 1, 2, and 3. Since $n \le 13$ we need at most $2^{12}=4,096$ vectors for ${\bf x}$. As we prefer to keep a non-standard generator matrix to distinguish Constructions I-IV, we have run all possibilities for ${\bf x}$ and have checked the equivalence of codes by Magma. Using our linux machine {\tt Intel(R) Xeon(R) CPU E3-1225 V2 @ 3.20GHz}, calculations for Theorems 1-3 were performed within ten minutes while some exhaustive search took more than two weeks. As future work, it is worth considering similar constructions for other finite fields and rings.

\medskip

\section*{Acknowledgements}
The author wants to thank Shitao Li for his careful reading.


\section*{Declarations}

\noindent
{\bf Conflict of interest}  The author declares that he has no conflict of interest regarding the publication of
this paper.


\end{document}